\documentclass[preprint,12pt]{elsarticle}

\usepackage{amssymb}
\usepackage{microtype}
\usepackage{graphicx}
\usepackage{subfigure}
\usepackage{booktabs} 
\usepackage{rotating}
\usepackage[title]{appendix}
\usepackage{gnuplottex}
\usepackage{nicefrac}
\usepackage{url}
\usepackage{xcolor}
\usepackage{amsmath}
\usepackage[numbers]{natbib}
\usepackage{amsthm}
\usepackage{soul}
\usepackage{geometry}
\usepackage{algorithm}
\usepackage[noend]{algpseudocode}
\usepackage{mathtools}
\usepackage{thm-restate}
\usepackage{afterpage}
\usepackage{rotfloat}
\usepackage{caption}
\usepackage{chngcntr}

\algnewcommand\algorithmicforeach{\textbf{for each}}
\algdef{S}[FOR]{ForEach}[1]{\algorithmicforeach\ #1\ \algorithmicdo}
\algrenewcommand\algorithmicrequire{\textbf{Input:}}
\algrenewcommand\algorithmicensure{\textbf{Output:}}

\DeclarePairedDelimiter\ceil{\lceil}{\rceil}
\DeclarePairedDelimiter\floor{\lfloor}{\rfloor}

\theoremstyle{plain}
\newtheorem{theorem}{Theorem}[section]

\newtheorem{lemma}[theorem]{Lemma}
\newtheorem{conjecture}[theorem]{Conjecture}
\newtheorem{corollary}[theorem]{Corollary}
\theoremstyle{definition}

\newtheorem{observation}[theorem]{Observation}
\theoremstyle{remark}

\journal{Computers and Operations Research}

\begin{document}

\begin{frontmatter}

\title{A $k$-swap Local Search for Makespan Scheduling\tnoteref{NWOFunding}}

\author[inst1]{Lars Rohwedder}
\ead{rohwedder@sdu.dk}
\author[inst2]{Ashkan Safari\corref{cor1}}
\ead{a.safari@maastrichtuniversity.nl}
\author[inst2]{Tjark Vredeveld}
\ead{t.vredeveld@maastrichtuniversity.nl}
\affiliation[inst1]{organization={Department of Mathematics and Computer Science, University of Southern Denmark},
            addressline={Moseskovvej}, 
            postcode={5230},
            city={Odense}, 
            country={Denmark}}
\affiliation[inst2]{organization={Department of Quantitative Economics, School of Business and Economics, Maastricht University},
            addressline={Tongersestraat 53}, 
            postcode={6211 LM},
            city={Maastricht}, 
            country={The Netherlands}}
\cortext[cor1]{Corresponding author}
\tnotetext[NWOFunding]{This research is supported by NWO grant OCENW.KLEIN.176. Moreover, this work was carried out while Lars Rohwedder was affiliated with Maastricht University.}

\begin{abstract}
Local search is a widely used technique for tackling challenging optimization problems, offering significant advantages in terms of computational efficiency and exhibiting strong empirical behavior across a wide range of problem domains.
In this paper, we address the problem of scheduling a set of jobs on identical parallel machines with the objective of \emph{makespan minimization}. For this problem, we consider a local search neighborhood, called \emph{$k$-swap}, which is a generalized version of the widely-used \emph{swap} and \emph{jump} neighborhoods.
The $k$-swap neighborhood is obtained by swapping at most $k$ jobs between two machines.
First, we propose an algorithm for finding an improving neighbor in the $k$-swap neighborhood which is faster than the naive approach, and prove an almost matching lower bound on any such an algorithm.
Then, we analyze the number of local search steps required to converge to a local optimum with respect to the $k$-swap neighborhood. 
For $k \geq 3$, we provide an exponential lower bound regardless of the number of machines, and for $k = 2$ (similar to the swap neighborhood), we provide a polynomial upper bound for the case of having two machines.
Finally, we conduct computational experiments on various families of instances.
\end{abstract}

\begin{keyword}
Local Search \sep Scheduling \sep Makespan Minimization \sep $k$-swap
\end{keyword}

\end{frontmatter}

\section{Introduction}
Local search methods are some of the most widely used heuristics for approaching computationally difficult optimization problems~\cite{aarts2003local, michiels2007theoretical}. Their appeal comes not only from their simplicity but also from their good empirical behavior as shown, for instance, by the \emph{2-opt} neighborhood~\cite{aarts2003local, croes1958method, flood1956traveling} for the traveling salesman problem (TSP).
The local search procedure starts with an initial solution and iteratively moves from a feasible solution to a \emph{neighboring} one until specific stopping criteria are satisfied. 
The performance of local search is significantly impacted by the selection of an appropriate neighborhood function. These functions determine the range of solutions that the local search procedure can explore in a single iteration.
The most straightforward type of local search, known as \emph{iterative improvement}, continuously selects a better solution within the neighborhood of the current one. The iterative improvement process stops when no further improvements can be made within the neighborhood, signifying that the procedure has reached a \emph{local optimum}.

In this paper, we consider one of the most fundamental scheduling problems, in which we are given a set of $n$ jobs, each of which needs to be processed on one of $m$ identical machines without preemption. All the jobs and machines are available from time $0$, and a machine can process at most one job at a time. 
The objective that we consider is \emph{makespan minimization}: we want the last job to be completed as early as possible. 
This problem, denoted by $P{\parallel}C_{\max}$~\cite{graham1979optimization}, is known to be strongly NP-hard~\cite{johnson1979computers}. In consequence, numerous studies in the literature have proposed approximation algorithms for this problem. As approximation algorithms are often considered impractical, some studies have approached this problem by local search, such as running time analysis of the iterative improvement procedure with respect to the \emph{jump} neighborhood, which was done by Brucker et al.~\cite{brucker1996improving, brucker1997improving}.
Despite the simplicity of this problem at first glance, Schuurman and Vredeveld~\cite{schuurman2007performance} pose an open question on the running time analysis of the iterative improvement procedure for the simplest generalization of the jump neighborhood to \emph{swap}. 
Accordingly, the limited theoretical study of local search algorithms for this problem motivates the study presented in this paper.

\paragraph{Related work} One of the simplest neighborhoods for addressing the problem $P{\parallel}C_{\max}$ is the so-called jump neighborhood, also known as the \emph{move} neighborhood, in which one job changes its machine allocation. Finn and Horowitz~\cite{finn1979linear} have proposed a simple improvement heuristic, called \textsc{0/1-Interchange}, which is able to find a local optimum with respect to the jump neighborhood. Brucker et al.~\cite{brucker1996improving, brucker1997improving} showed that the \textsc{0/1-Interchange} algorithm terminates in $O(n^2)$ steps, regardless of the job selection rule; the tightness of this bound was shown by Hurkens and Vredeveld~\cite{hurkens2003local}. 
The upper bound of $O(n^2)$ was first obtained for schedules with two machines and was extended later to schedules with an arbitrary number of machines.
Finn and Horowitz~\cite{finn1979linear} gave a guarantee of $2 - \frac{2}{m+1}$ on the quality of the local optimum with respect to the jump neighborhood.
Another well-known folklore neighborhood is the swap neighborhood in which two jobs interchange their machine allocations.
If all jobs are scheduled on the same machine, then no swap neighbor exists; therefore, in the literature~\cite{schuurman2007performance}, the swap neighborhood is defined as one that consists of all possible jumps and swaps.
The guarantee of $2 - \frac{2}{m+1}$ also holds on the quality of the local optimum with respect to the swap neighborhood; the tightness of this bound was shown by Schuurman and Vredeveld~\cite{schuurman2007performance}.
Other neighborhoods, such as \emph{push}~\cite{schuurman2007performance}, \emph{multi-exchange}~\cite{ahuja2001multi, frangioni2004multi} and \emph{split}~\cite{brueggemann2011exponential} have also been proposed in the literature to solve this problem.
In addition to local search, various approximation algorithms have been introduced to address scheduling problems. Among them, some of the first approximation algorithms~\cite{graham1966bounds, graham1969bounds} and polynomial time approximation schemes (PTAS)~\cite{hochbaum1987using} have been developed in the context of scheduling problems.
It has been shown that there exists a fully polynomial time approximation scheme (FPTAS) for this problem~\cite{sahni1976algorithms, woeginger2000does} when the number of machines is a constant $m$. 
When the number of machines is part of the input, a PTAS has been provided for this problem~\cite{hochbaum1987using}, and the best-known running time of such a scheme is proposed by Berndt et al.~\cite{berndt2022load}.
PTAS's are often considered impractical due to their tendency to yield excessively high running time bounds (albeit polynomial) even for moderate values of $\varepsilon$ (see~\cite{marx2008parameterized}). In this paper, we focus on a local search study of the problem $P||C_{\max}$.
Local search methods are studied extensively for various other problems as well~\cite{aarts2003local}. In particular, the \emph{$k$-opt} heuristic has been proposed for TSP, and has been the subject of theoretical study. For example, de Berg et al.~\cite{buchin55fine} have proposed an algorithm that finds the best \emph{$k$-opt move} faster than the naive approach.

The scheduling problem studied in this paper has applications in various domains. For instance, in manufacturing systems, it is used to assign tasks to machines in a way that minimizes the makespan, ensuring efficient production scheduling~\cite{pinedo2012scheduling}. In cloud computing, the problem arises in optimizing the allocation of tasks to virtual machines to balance the load and minimize the execution time~\cite{calheiros2011cloudsim, rosa2024empowering}. In distributed systems, it plays a key role in load balancing to maximize resource utilization and minimize processing delays~\cite{hanamakkanavar2015load, wang2023distributed}. In satellite scheduling, it helps efficiently allocate tasks such as communication or data collection across satellites to enhance mission performance and reduce idle time~\cite{wang2024adaptive, wu2022ensemble, wu2024improved}. 

Since the problem of minimizing the makespan on identical parallel machines is strongly NP-hard~\cite{johnson1979computers}, it is unlikely that there exists an efficient algorithm that solves the problem optimally~\cite{dell2008heuristic}; several exact algorithms for this problem (see e.g.~\cite{mokotoff2004exact}) are often based on solving a (mixed) integer program and require a lot of computational power.
On the other hand, while numerous studies in the literature have explored approximation algorithms for addressing this problem~\cite{hochbaum1987using, sahni1976algorithms, woeginger2000does}, these approaches often suffer from high computational complexity, making them less practical for large-scale real-world applications~\cite{marx2008parameterized}. 
These limitations motivated us to study the problem using local search, which not only exhibits good empirical performance but is also easy to understand and implement, making it useful for practical applications.

\begin{figure}[t]
\centering
\includegraphics[width=0.6\textwidth]{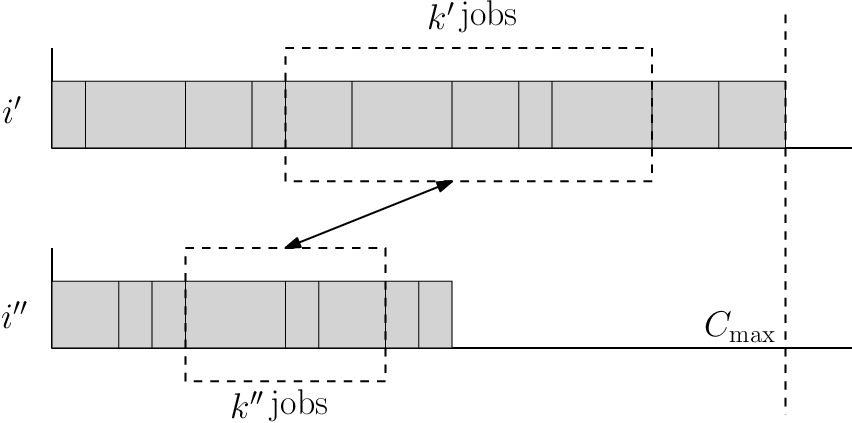}
\caption{$k$-swap operator.}
\label{fig:k-swapOerator}
\end{figure}

\paragraph{Formal model} As the order in which the jobs, assigned to the same machine, are processed, does not influence the makespan, we represent a schedule by $\sigma = (M_1, M_2, \dotsc , M_m)$, where $M_i$ represents the set of jobs scheduled on machine $i$ (for $i \in \{1, 2, \dotsc , m\}$).
Let $L_i = \sum_{j \in M_i} p_j$ denote the load of machine $i$, where $p_j$ is the time required to process job $j$. In this setting, the makespan of a schedule, commonly denoted by $C_{\max}$, is equal to the maximum load, i.e., $L_{\max} = \max_{i} L_i$, which we want to minimize (we use the notations $C_{\max}$ and $L_{\max}$ interchangeably throughout the paper). A \emph{critical} machine is a machine whose load is equal to the makespan, and a \emph{non-critical} machine is a machine with a load less than $L_{\max}$. 
Let $\mathcal{M}_{\max} = \{ i: L_i = L_{max} \}$ denote the set of critical machines.
Moreover, we define $\Delta = L_{\max} - L_{\min}$, as the difference between the maximum and minimum machine load, where $L_{\min} = \min_{i} L_i$.

In this paper, we consider a generalization of the jump and swap neighborhoods, which is called the \emph{$k$-swap} neighborhood.
To obtain a $k$-swap neighbor, we select $k'$ jobs from a machine $i'$ and $k''$ jobs from a machine $i''$, such that $k' + k'' \leq k$, and then we interchange the machine allocations of these jobs~(see Fig.~\ref{fig:k-swapOerator}). 
We say we are in a \emph{k-swap optimal solution} if we cannot decrease the makespan or the number of critical machines, while keeping the makespan equal, by applying a $k$-swap operator. For $k = 1$, the $k$-swap neighborhood is the same as the jump neighborhood, and for $k = 2$, the $k$-swap neighborhood contains all possible jumps and swaps.

\paragraph{Our contribution}
A naive implementation of the $k$-swap takes $O(k \cdot n^k)$ steps to find an improving neighbor or to determine that no improving neighbor exists.
In this paper, we propose a randomized algorithm for finding an improving solution in the $k$-swap neighborhood in $O(n^{\ceil{\frac{k}{2}} + O(1)})$ steps, if one exists, with the probability of success of at least $\dfrac{\binom{k}{\ceil{\frac{k}{2}}}}{2^k}$. By repeating the algorithm $\Theta(\sqrt{k})$ times, we can increase the probability of success to at least $1 - \frac{1}{e}$.

\begin{restatable}{theorem}{randomizedtheorem}\label{thrm:complexityOfTheAlg}
    We are able to find a better solution in the k-swap neighborhood with the probability of success of at least $1 - \frac{1}{e}$ in $O(m^2 \cdot k^{3/2} \cdot n^{\ceil{\frac{k}{2}}} \cdot \log n)$ by running Algorithm~\ref{Alg:k-Swap}, for every possibility of swapping exactly $\mathcal{K} = 1, \dotsc , k$ jobs, $\Theta(\sqrt{k})$ times.
\end{restatable}

the \noindent Furthermore, we show how to derandomize the above-mentioned algorithm, resulting in an algorithm that finds an improving neighbor, if one exists.

\begin{restatable}{theorem}{derandomizedtheorem}\label{cor:deRandomizedTimeComp}
    The derandomized neighborhood search terminates in $O(m^2 \cdot k^{9} \cdot n^{\ceil{\frac{k}{2}}} \cdot \log n \cdot \log k)$ steps.
\end{restatable}

\noindent To establish a lower bound on the running time of finding an improving neighbor in the $k$-swap neighborhood, if one exists, we make use of the $k$-sum conjecture~\cite{abboud2013exact}. The conjecture states that for every $k \geq 2$ and $\varepsilon > 0$, the $k$-sum problem, which asks whether there exist $k$ distinct numbers in a given set of $q$ numbers that sum to zero, cannot be solved in $O(q^{\lceil k/2 \rceil - \varepsilon})$ (randomized) time.
The $k$-sum problem can be solved in time $O(q^{\lceil k/2 \rceil})$, and it has been shown that the 3-sum problem is solvable in time $O(\frac{q^2}{\log^2 q})$~\cite{baran2008subquadratic}.
This conjecture has been used to provide tight lower bounds for many problems in computational geometry~\cite{barequet2001polygon, barrera1996finding, gajentaan1995class} as well as some discrete problems~\cite{abboud2014popular, jafargholi2016mathrm, patrascu2010towards, vassilevska2009finding}.

\begin{restatable}{theorem}{ksumlowerbound}\label{thrm:kSumLowerBound}
    For any $k \geq 2$ and $\varepsilon > 0$, there does not exist a randomized algorithm that finds an improving neighbor with high probability, if one exists, in time $O(n^{\lceil k/2 \rceil - \varepsilon})$, unless the $k$-sum conjecture is false.
\end{restatable}

\noindent We complement these results by analyzing the number of local search steps required to converge to a local optimum with respect to the $k$-swap neighborhood.
For $k=2$, an upper bound of $O(n^4)$ is provided for the case of having only two machines.

\begin{restatable}{theorem}{twoswapupper}\label{theorem:2-swapOptSolUpperBound}
    The total number of improving iterations in finding a 2-swap optimal solution is upper-bounded by $O(n^4)$ for the case of having $m=2$ machines.
\end{restatable}

\noindent However, for $k \geq 3$, we establish an exponential lower bound regardless of the number of machines.
Establishing such a lower bound is based on a sequence of 3-swaps which takes exponential time in the number of jobs. 

\begin{restatable}{theorem}{threeswaplower}\label{theorem:lowerbound}
    There exists an instance and an initial schedule as well as a sequence of 3-swaps such that a 
    $k$-swap optimal solution, for $k \geq 3$, is reached after $2^{\Omega(n)}$ iterations.
\end{restatable}

\noindent Finally, we have conducted computational experiments on various families of instances, and based on these experiments, we conjecture if the exponential lower bound for $k \geq 3$ represents an overly pessimistic scenario and that such a situation may be rare in practice.

\section{Searching the Neighborhood\label{sec:SearchNeighborhood}}
In this section, we present an algorithm for finding an improving solution in the $k$-swap neighborhood, which runs faster than the naive approach.
In a naive approach, we consider swapping all different combinations of at most $k$ jobs with each other. As we have $n$ jobs in our schedule, we need to search over all $O(n^k)$ different combinations in the worst case to find an improvement or to determine that no improving neighbor exists. In our algorithm, we consider all combinations of at most $\frac{k}{2}$ jobs and then, combine two such combinations to one $k$-swap.
This principle is known as the \emph{meet-in-the-middle} approach~\cite{horowitz1974computing}. The proposed algorithm is a randomized \emph{Monte Carlo} algorithm with \emph{one-sided error}~\cite{motwani1995randomized}. More precisely, when this algorithm outputs an improving solution in the $k$-swap neighborhood, this solution is always correct. Conversely, if the algorithm does not find an improvement, there is a nonzero probability that an improvement still exists.

Before we describe the algorithm, recall that a schedule is represented by $\sigma = (M_1, \dotsc , M_m)$, where $M_i$ denotes the set of jobs scheduled on machine $i$.
The load of machine $i$ is denoted by $L_i = \sum_{j \in M_i} p_j$, where $p_j$ is the time required to process job $j$.
Moreover, $\mathcal{M}_{\max}$ denotes the set of critical machines.
An improving $k$-swap neighbor is obtained by swapping a set $S_{i'} \subseteq M_{i'}$ of $k'$ jobs on a critical machine $i' \in \mathcal{M}_{\max}$ with a set $S_{i''} \subseteq M_{i''}$ of $k''$ jobs on a non-critical machine $i'' \notin \mathcal{M}_{\max}$ such that $k' + k'' \leq k$ and $0 < \sum_{j \in S_{i'}} p_j - \sum_{j \in S_{i''}} p_j < L_{i'} - L_{i''}$.

To describe the algorithm, for the sake of simplicity, we consider swapping exactly $k$ jobs with each other. We can run the proposed algorithm $k-1$ more times for the case of swapping less than $k$ jobs.
In this algorithm, first, we divide the set of jobs on a critical machine $i'$ and a non-critical machine $i''$ into two sets $A$ and $B$ uniformly at random (we will see that, with large enough probability, $\frac{k}{2}$ out of $k$ jobs for which we are looking, will be assigned to $A$ and the remaining $\frac{k}{2}$ jobs to $B$).
Then, we compute the sets $A_{\Sigma} = \{ \sum_{j \in S_A \cap M_{i'}} p_j - \sum_{j \in S_A \cap M_{i''}} p_j : S_A \subseteq A, |S_A| = \ceil{\frac{k}{2}} \}$ and $B_{\Sigma} = \{ \sum_{j \in S_B \cap M_{i'}} p_j - \sum_{j \in S_B \cap M_{i''}} p_j: S_B \subseteq B, |S_B| = \floor{\frac{k}{2}} \}$, where these two sets are all possible changes in machine loads of swapping exactly $\frac{k}{2}$ jobs assigned to the sets $A$ and $B$, respectively.
If there exist $x \in A_{\Sigma}$ and $y \in B_{\Sigma}$ such that $0 < x + y < L_{i'} - L_{i''}$, there exists an improving $k$-swap. To determine the corresponding $k$-swap, we need to know which jobs are used to find the values $x$ and $y$. Therefore, we should also store references to the corresponding jobs of each value in $A_{\Sigma}$ and $B_{\Sigma}$, but for simplicity of exposition, we leave out this reference in the remainder of the section.
To find values $x \in A_{\Sigma}$ and $y \in B_{\Sigma}$ satisfying the above criteria, we sort the values in the set $B_{\Sigma}$ in non-decreasing order.
Then, for every $x \in A_{\Sigma}$, we search for some $y$ in the sorted list of the values in $B_{\Sigma}$, such that $y \in (-x, L_{i'} - L_{i''} - x)$.
This can be implemented by binary search.
As long as no improving $k$-swap has been found, we repeat the above procedure for every combination of critical machine $i'$ and non-critical machine $i''$.
Our algorithm, which we refer to as \emph{randomized neighborhood search}, is summarized in Algorithm~\ref{Alg:k-Swap}.

\begin{algorithm}[ht]
\caption{Randomized Neighborhood Search}\label{Alg:k-Swap}

\begin{algorithmic}[1]
\Require A schedule $\sigma = (M_1, M_2, \dotsc , M_m)$ and a value $k$.
\Ensure A Boolean variable, \{True, False\}, which indicates whether $\sigma$ has been improved or not.

\ForEach {$i' \in \mathcal{M}_{\max}$ and $i'' \notin \mathcal{M}_{\max}$}

\State Initialize $A = B = A_{\Sigma} = B_{\Sigma} = \varnothing$
\State {Assign each $j \in M_{i'} \cup M_{i''}$ uniformly at random to $A$ or $B$}

\State {$A_{\Sigma} = \{ \sum_{j \in S_A \cap M_{i'}} p_j - \sum_{j \in S_A \cap M_{i''}} p_j : S_A \subseteq A, |S_A| = \ceil{\frac{k}{2}} \}$}

\State {$B_{\Sigma} = \{ \sum_{j \in S_B \cap M_{i'}} p_j - \sum_{j \in S_B \cap M_{i''}} p_j: S_B \subseteq B, |S_B| = \floor{\frac{k}{2}} \}$}

\State {Sort the values in $B_{\Sigma}$ in non-decreasing order}
\ForEach {$x \in A_{\Sigma}$}
\State {Use binary search to find a value $y \in B_{\Sigma}$ such that $y \in (-x, L_{i'} - L_{i''} - x)$}
\If {such a value $y$ exists}
\State{Update $\sigma$ by swapping the corresponding jobs to $x$ and $y$}
\State{\textbf{return} True}
\EndIf
\EndFor
\EndFor
\State{\textbf{return} False}
\end{algorithmic}
\end{algorithm}

The remainder of this section is structured as follows. In Subsection~\ref{subsec:probOfSuccess}, we analyze the probability that Algorithm~\ref{Alg:k-Swap} finds an improving $k$-swap neighbor if one exists, and also determine the number of repetitions to be made, to make this probability large enough.
In Subsection~\ref{subsec:RTALG1}, we analyze the running time of Algorithm~\ref{Alg:k-Swap}. 
In Subsection~\ref{subsec:derand}, we mention how the algorithm can be derandomized. 
Finally, in Subsection~\ref{subsec:LBKSwap}, we establish a lower bound on the running time to find an improving $k$-swap neighbor.

\subsection{Probability of Success\label{subsec:probOfSuccess}}
As mentioned, the proposed algorithm has a one-sided error as even when an improving $k$-swap neighbor exists, the algorithm might not find one.
If this is the case, we say that the algorithm \emph{fails}; in case it does find an improving neighbor, we say that the algorithm is \emph{successful}.
Suppose that there exists an improving $k$-swap swapping exactly $k$ jobs. Then, Algorithm~\ref{Alg:k-Swap} is successful if exactly $\ceil{\frac{k}{2}}$ of these jobs are assigned to set $A$ and the remaining $\floor{\frac{k}{2}}$ jobs are assigned to set $B$.
As there are $\binom{k}{\ceil{\frac{k}{2}}}$ ways of assigning $\frac{k}{2}$ out of $k$ jobs to the set $A$ and the remaining jobs to the set $B$, the probability of success of the algorithm 
is at least $\dfrac{\binom{k}{\ceil{\frac{k}{2}}}}{2^k}$.
This implies that the probability of failure of the algorithm is at most $1 - \dfrac{\binom{k}{\ceil{\frac{k}{2}}}}{2^k}$. 
Repeating Algorithm~\ref{Alg:k-Swap} $\dfrac{2^k}{\binom{k}{\ceil{\frac{k}{2}}}}$ times lead to the following lemma.

\begin{lemma}
\label{lem:probFailureOfKSwap}
By repeating the randomized neighborhood search (Algorithm~\ref{Alg:k-Swap}) $\dfrac{2^k}{\binom{k}{\ceil{\frac{k}{2}}}}$ times, the probability of failure in finding an improving schedule can be bounded from above by $\frac{1}{e}$.
\end{lemma}

\subsection{Running Time\label{subsec:RTALG1}}
For given machines $i' \in \mathcal{M}_{\max}$ and $i'' \notin \mathcal{M}_{\max}$, determining the values for sets $A_{\Sigma}$ and $B_{\Sigma}$ can be done in $O(n^{\ceil{\frac{k}{2}}})$ time and sorting the set $B_{\Sigma}$ in time $O(k \cdot n^{\floor{\frac{k}{2}}} \cdot \log n)$. In the worst case, every element $x \in A_{\Sigma}$ needs to be checked, and the binary search over the set $B_{\Sigma}$ takes $O(k \log n)$ time, resulting in the worst-case time complexity of $O(k \cdot n^{\ceil{\frac{k}{2}}} \cdot \log n)$.
As we consider all the combinations of critical and non-critical machines, the above procedure is repeated $O(m^2)$ times, leading to the following lemma.

\begin{lemma}
\label{lem:timeComplexityOfKSwap}
The randomized neighborhood search (Algorithm~\ref{Alg:k-Swap}) terminates in $O(m^2 \cdot k \cdot n^{\ceil{\frac{k}{2}}} \cdot \log n)$ steps.
\end{lemma}

To obtain a success probability of at least $1 - \frac{1}{e}$, we repeat the randomized neighborhood search $\dfrac{2^k}{\binom{k}{\ceil{\frac{k}{2}}}}$ times (see Lemma~\ref{lem:probFailureOfKSwap}).
As $\binom{2n}{n} = \Theta\left(\dfrac{4^n}{\sqrt{\pi n}}\right)$, according to the asymptotic growth of the $n$th central binomial coefficient~\cite{luke1969special}, we get 
\begin{align*}
    \binom{k}{\ceil{\frac{k}{2}}} = \binom{k}{\floor{\frac{k}{2}}} =  \Theta \binom{2\floor{\frac{k}{2}}}{\floor{\frac{k}{2}}} = \Theta \left(\frac{4^{\floor{\frac{k}{2}}}}{\sqrt{\pi \floor{\frac{k}{2}}}}\right) = \Theta\left(\frac{2^k}{\sqrt{k}}\right).
\end{align*}
Therefore, we have the following observation.

\begin{observation}
\label{obs:stirling}
$\dfrac{2^k}{\binom{k}{\ceil{\frac{k}{2}}}} = \Theta(\sqrt{k})$.
\end{observation}
In Algorithm~\ref{Alg:k-Swap}, we consider swapping exactly $k$ jobs. However, for obtaining a $k$-swap neighbor, we swap at most $k$ jobs. 
Therefore, we repeat Algorithm~\ref{Alg:k-Swap} to consider swapping exactly $\mathcal{K} = 1, \dotsc, k$ jobs.
As we have
\begin{align*}
    k \cdot n^{\frac{k}{2}} + (k-1) \cdot n^{\frac{k-1}{2}} + \dotsc + 1 \cdot n^{\frac{1}{2}}  
    \leq k \cdot (n^{\frac{k}{2}} + n^{\frac{k-1}{2}} + \dotsc + n^{\frac{1}{2}}),
\end{align*}
and $n^{\frac{k}{2}} + n^{\frac{k-1}{2}} + \dotsc + n^{\frac{1}{2}} = O(n^{\frac{k}{2}})$, the time complexity of finding a $k$-swap neighbor remains $O(m^2 \cdot k \cdot n^{\ceil{\frac{k}{2}}} \cdot \log n)$.
By Lemma~\ref{lem:probFailureOfKSwap} and Observation~\ref{obs:stirling}, we know that we need to repeat the algorithm $\Theta(\sqrt{k})$ times to obtain a success probability of at least $1 - \frac{1}{e}$, which leads to the following theorem.

\randomizedtheorem*

\subsection{Derandomization\label{subsec:derand}}
In this subsection, we describe a method for derandomizing the proposed randomized algorithm. 
This method is based on a combinatorial structure called \emph{splitters} which was introduced by Naor et al.~\cite{naor1995splitters} to derandomize the \emph{color coding} technique, introduced by Alon et al.~\cite{alon1995color}. Formally, an ($n, k, l$)-splitter is a family of hash functions $H$ from $\{1, ... , n\}$ to $\{1, ... , l\}$ such that for all $S \subseteq \{1, ... , n\}$ with $|S| = k$, there exists a function $h \in H$ that splits $S$ evenly; i.e., for every $j, j' \in \{1, ..., l\}$, $|h^{-1}(j) \cap S|$ and $|h^{-1}(j') \cap S|$ differ by at most 1. It has been shown~\cite{alon1995color, naor1995splitters} that there exists an ($n, k, k^2$)-splitter of size $k^6 \cdot \log k \cdot \log n$ which is computable in time complexity of $O(k^6 \cdot \log k \cdot n \cdot \log n)$.

In our randomized algorithm, we divide the set of jobs uniformly at random into two sets $A$ and $B$. 
If there exists a set of $k$ jobs such that by swapping them we are able to improve the schedule, the algorithm is successful if $\ceil{\frac{k}{2}}$ of these jobs are in the set $A$ and the remaining are in the set $B$. In the derandomized version of this algorithm, instead of dividing the set of jobs uniformly at random into two sets $A$ and $B$, we divide the set of jobs in several different ways based on the functions in an ($n, k, k^2$)-splitter. We will show that based on at least one of the functions, $\ceil{\frac{k}{2}}$ of the jobs, for which we are looking, will be assigned to the set $A$, and the remaining will be assigned to the set $B$.
The idea of how to derandomize our algorithm is explained in more detail below.

If we map each of the jobs to a number $\{ c_1, c_2, ... , c_{k^2} \}$ based on the functions in an ($n, k, k^2$)-splitter, according to at least one of these functions, those $k$ jobs, for which we are looking, will be mapped to different numbers. As we want $\ceil{\frac{k}{2}}$ of these $k$ jobs to be assigned to the set $A$ and the remaining to be assigned to the set $B$, we map the numbers $\{ c_1, c_2, ... , c_{k^2} \}$ to only two numbers 0 and 1, and then, we say those jobs which are mapped to 0 are in the set $A$ and the jobs which are mapped to 1 are in the set $B$ (or vice versa). 
To make sure that after mapping the numbers $\{ c_1, c_2, ... , c_{k^2} \}$ to 0 and 1, $\ceil{\frac{k}{2}}$ of these jobs are mapped to 0 (or 1) and the rest are mapped to 1 (or 0), we consider $2k^2$ different mappings $\{ m_1, m_2, ... , m_{2k^2} \}$ of the numbers $\{ c_1, c_2, ... , c_{k^2} \}$ to 0 and 1 as follows.
In mapping $m_1$, we map the numbers $\{ c_1, c_2, ... , c_{\ceil{\frac{k^2}{2}}} \}$ to 0 and the rest to 1.
Mapping $m_i$, where $i > 1$ is even, is obtained from $m_{i-1}$ by mapping the number $c_{\ceil{\frac{k^2}{2}} + \frac{i}{2}}$ to 0.
Mapping $m_i$, where $i > 1$ is odd, is obtained from $m_{i-1}$ by mapping the number $c_{\floor{\frac{i}{2}}}$ to 1. Intuitively, we are doing circular shifts until we obtain a mapping which is the complement of $m_1$ (see Fig. \ref{fig:mappings}).
In the following lemma, we show that in at least one of these mappings, $\ceil{\frac{k}{2}}$ of the jobs, for which we are looking, are colored with 0 and the rest are colored with 1.

\begin{figure}[t]
\centering
\begin{equation*}
\begin{matrix}
    & c_1 & c_2 & \cdots \!\! & c_{\lceil \frac{k^2}{2}\rceil} \!\!\!\!  & c_{\lceil \frac{k^2}{2}\rceil + 1} \!\!\!\! & c_{\lceil \frac{k^2}{2}\rceil + 2} \!\! & \cdots  & c_{k^2 - 1} \!\! & c_{k^2} \\\hline
    m_1 & 0 & 0 & \cdots & 0 & 1 & 1 & \cdots & 1 & 1 \\
    m_2 & 0 & 0 & \cdots & 0 & 0 & 1 & \cdots & 1 & 1 \\
    m_3 & 1 & 0 & \cdots & 0 & 0 & 1 & \cdots & 1 & 1 \\
     & \vdots &  & \ddots &  & \vdots & & \ddots &  & \vdots \\
    m_{2k^2} & 1 & 1 & \cdots & 1 & 0 & 0 & \cdots & 0 & 0
\end{matrix}    
\end{equation*}
\caption{Mappings of the numbers $\{ c_1, c_2, ... , c_{k^2} \}$ to 0 and 1.}
\label{fig:mappings}
\end{figure}

\begin{lemma}
\label{lem:derandomizationMapping}
Let $\mathcal{C}$ be a set of $k$ jobs that are mapped to different numbers from the set $\{ c_1, c_2, ... , c_{k^2} \}$. In at least one of the mappings $\{ m_1, m_2, ... , m_{2k^2} \}$, $\ceil{\frac{k}{2}}$ of the jobs in $\mathcal{C}$ are colored with 0 and the rest are colored with 1.
\end{lemma}
\begin{proof}
Let $A_i$ be the set of the jobs in $\mathcal{C}$ that are colored with 0 in mapping $m_i$. We have $0 \leq |A_i| \leq k$.
As we have $|A_{2k^2}| = k - |A_{1}|$, and for $i \in \{ 2, \dotsi, 2k^2 \}$ we have $|A_{i-1}| - 1 \leq |A_i| \leq |A_{i-1}| + 1$ (since each mapping $m_i$ is obtained from mapping $m_{i-1}$ by changing the value of only one of the numbers), in one of the mappings we get $|A_i| = \ceil{\frac{k}{2}}$.
\end{proof}

As it has been shown that there exists an ($n, k, k^2$)-splitter of size $k^6 \cdot \log k \cdot \log n$, and we showed that for each function we need to consider $2k^2$ different ways of mapping $k^2$ numbers to only two numbers, we get the following corollary.

\begin{corollary}
\label{cor:deRandomization}
By running the randomized neighborhood search algorithm $O(k^8 \cdot \log k \cdot \log n)$ times based on an ($n, k, k^2$)-splitter, which is computable in time $O(k^6 \cdot \log k \cdot n \cdot \log n)$, we are able to derandomize it.
\end{corollary}

Moreover, according to Lemma~\ref{lem:timeComplexityOfKSwap} and Corollary~\ref{cor:deRandomization}, we have the following theorem.

\derandomizedtheorem*

\subsection{Lower Bound on the Running Time\label{subsec:LBKSwap}}
\label{sec:Lower Bound on the Running Time}
In this part, we establish a lower bound for the problem of finding an improving solution in the $k$-swap neighborhood by reducing the $k$-sum problem to the decision version of this problem. 
In the $k$-sum problem, we are given a set of $q$ integer numbers $S = \{a_1, a_2, \dotsc , a_q\}$, and we need to decide if there are exactly $k$ distinct numbers in $S$ that sum to zero.
Providing this lower bound is based on the following hypothesis, which is called the $k$-sum conjecture~\cite{abboud2013exact}.

\begin{conjecture}
There does not exist any $k \geq 2$, an $\varepsilon > 0$, and a randomized algorithm that succeeds with high probability in solving $k$-sum in time $O(q^{\lceil k/2 \rceil - \varepsilon})$.
\end{conjecture}

\ksumlowerbound*
\begin{proof}
We reduce the $k$-sum problem to the decision version of the problem of finding an improving solution in the $k$-swap neighborhood.
Hereto, given an instance for the $k$-sum problem, we need to define a scheduling instance and a starting schedule such that an improving $k$-swap neighbor exists if and only if we have a yes-instance for the $k$-sum problem.
Consider an instance for the $k$-sum problem on a set $S = \{a_1, a_2, \dotsc , a_q\}$, and assume w.l.o.g. that $\sum_{j: a_j \in S} a_j \geq 0$.
We construct a scheduling instance on $m=2$ machines with $n = q + k + 1$ jobs.
The first $q$ jobs, called the \emph{$k$-sum jobs}, have processing time $p_j = q \cdot a_j + 1$ if $a_j \geq 0$ and $p_j = -q \cdot a_j$ whenever $a_j < 0$. 
Jobs $j = q+1, \ldots, q+k-1$ have processing time $p_j = 3 \cdot q \cdot \sum_{j:a_j \in S} |a_j|$, and the remaining two jobs have processing times $p_{q+k} = 3 \cdot q \cdot \sum_{j:a_j \in S} |a_j| + 1$ and $p_{q+k+1} = q \cdot (\sum_{j: a_j \in S} a_j + 3 \cdot k \cdot \sum_{j:a_j \in S} |a_j|)$.
The starting schedule is obtained by assigning all $k$-sum jobs with $a_j \geq 0$ as well as jobs $q+1, \ldots, q+k$ to machine~1 and all $k$-sum jobs with $a_j < 0$ and job $q+k+1$ to machine~2. Then, $L_1  - L_2 = \theta + 1$, where $\theta$ is the number of $k$-sum jobs with $a_j \geq 0$.

Suppose that there is a set $\Gamma \subseteq S$ of distinct numbers such that $\sum_{j:a_j \in \Gamma} a_j = 0$. A set $U_1$ of the corresponding jobs of these numbers are in $M_{1}$ and a set $U_2$ are in $M_{2}$, such that $|U_1| + |U_2| = |\Gamma|$.  We denote the summation of the processing times of the jobs in $U_1$ and $U_2$ by $p(U_1)$ and $p(U_2)$, respectively. Note that $p(U_1) > p(U_2)$. Moreover, let $A$ and $B$ denote the summation of the absolute values of positive and negative numbers in $\Gamma$, respectively.  Since we have $A - B = 0$, we get
\begin{align*}
   p(U_1) - p(U_2) = q \cdot (A - B) + |U_1| = |U_1| \leq \theta < L_1 - L_2.  
\end{align*}
As we have $0 < p(U_1) - p(U_2) < L_1 - L_2$, we are able to improve our schedule by swapping the jobs in $U_1$ with the jobs in $U_2$.

To show that an improving $k$-swap neighbor implies that the corresponding instance for $k$-sum is a yes-instance, suppose that there are a set $U_1 \subseteq M_{1}$ and a set $U_2 \subseteq M_{2}$ of jobs such that $0 < p(U_1) - p(U_2) < L_1 - L_2$ and $|U_1| + |U_2| = k$. Since $L_1 - L_2 = \theta + 1 \leq q + 1$,  all the jobs in $U_1 \cup U_2$ are $k$-sum jobs.
More precisely, if we have a job $j \in U_1$ with $p_j \geq 3 \cdot q \cdot \sum_{j:a_j \in S} |a_j|$, we get $p(U_1) - p(U_2) \geq 2q$, when job $q+k+1$ is not in $U_2$, and if we have job $q+k+1$ in $U_2$, we get $p(U_1) - p(U_2) < 0$.
Let $U_1^S \subseteq S$ and $U_2^S \subseteq S$ denote the set of corresponding numbers in $S$ to $k$-sum jobs in $U_1$ and $U_2$, respectively.
As all the jobs in $U_1 \cup U_2$ are $k$-sum jobs, We have
\begin{align*}
p(U_1) - p(U_2) = q \cdot d + |U_1|,
\end{align*}
where 
$d = \sum_{j: a_j \in U_1^S} |a_j| - \sum_{j: a_j \in U_2^S} |a_j|$. Since we have, $0 < p(U_1) - p(U_2) < L_1 - L_2$, $d$ must be equal to zero. Since we have $d = 0$, there exist $k$ distinct numbers in $S$ that sum to zero.
\end{proof}

\section{Running Time of Finding a Local Optimum}
So far, we have focused on finding an improving neighbor in the $k$-swap neighborhood. In this section, we consider the number of iterations it takes to converge to a local optimum. First, in Subsection~\ref{subsec:2SwapRunningTime}, we consider the regular swap ($k=2$) and provide an upper bound of $O(n^4)$ for the case of having two machines in our schedule. 
Next, in Subsection~\ref{subsec:3SwapRunningTime}, we provide an instance and a starting solution for which there exists a sequence of improving $3$-swaps that has an exponential length, even on $m=2$ machines.

\subsection{Obtaining a 2-swap Optimal Solution\label{subsec:2SwapRunningTime}}
In this part, we restrict ourselves to the case in which there are only $m=2$ machines.
To provide an upper bound on the number of iterations required in obtaining a 2-swap optimal solution, we first give a bound on the number of times that the machine that is critical changes. Then, by bounding the number of iterations that a certain machine can remain critical after a 2-swap, we can bound the number of iterations until we have found a 2-swap optimal solution.

\begin{lemma}
\label{lem:remainingCritical}
A critical machine remains critical for at most $O(n^2)$ iterations in the procedure of obtaining a 2-swap optimal solution whenever the number of machines is restricted to $m=2$.
\end{lemma}
\begin{proof}
W.l.o.g., we assume that $p_1 \leq p_2 \leq \dotsc \leq p_n$, and machine 1 is the critical machine, i.e., $L_1 = L_{\max}$. Define $\Phi = \sum_{j \in M_1} \mathrm{rank}(j)$, where $\mathrm{rank}(j) = j$, for $j \in \{1, 2, \dotsc, n\}$. Then, $\Phi \leq n^2$.
In an improving $2$-swap, in which we swap a job $j \in M_1$ with a job $j' \in M_2$, we have that $p_j > p_{j'}$ and thus the value of $\Phi$ decreases. Similarly, when the $2$-swap is obtained by moving one job or two jobs from machine 1 to machine 2, $\Phi$ will also decrease.
As the values of $\Phi$ are always decreasing, the critical machine changes after at most $O(n^2)$ iterations.
\end{proof}

In the next lemma, we provide an upper bound on the number of times that the critical machine changes. The idea of the proof of the following lemma is similar to Brucker et al.~\cite{brucker1996improving}, who provided an upper bound on the number of jumps in converging to a local optimum for the case of $m=2$ machines.

\begin{lemma}
\label{lem:changingCritical}
The critical machine changes at most $O(n^2)$ times in the procedure of obtaining a 2-swap optimal solution whenever the number of machines is restricted to $m=2$.
\end{lemma}
\begin{proof}
As there are only two machines, we know that $\Delta = |L_1 - L_2|$. Moreover, we know that the value of $\Delta$ is strictly decreasing over the iterations as the load of the critical machine is decreasing whereas the load of the other machine is increasing.
Consider an iteration $t$ in which a 2-swap is performed and assume that the critical machine changes. Assume w.l.o.g. that before performing this 2-swap, machine 1 is the critical machine. Let $\delta < \Delta$ be the amount by which the load of machine 1 decreases, and thus the load of machine 2 increases. As the critical machine changes, we know that $L_1 - \delta < L_2 + \delta$.
For the difference of machine loads after the 2-swap, we have $\Delta' = L_2 + \delta - (L_1 - \delta) = 2 \delta - \Delta = \delta - (\Delta - \delta)$. Using the fact that $\delta < \Delta$, we get $\Delta' < \delta$. 
As $\Delta$ is decreasing, after iteration $t$, any 2-swap in which the load of the critical machine decreases by at least $\delta$, will increase the makespan and therefore will not be performed. In particular, the 2-swap that was done in iteration $t$ cannot be performed anymore. As there are at most $O(n^2)$ possible 2-swaps, the number of iterations in which the critical machine changes is bounded by $O(n^2)$.
\end{proof}

According to Lemma~\ref{lem:remainingCritical} and Lemma~\ref{lem:changingCritical}, we get the following theorem.

\twoswapupper*

\subsection{An Exponential Lower Bound for $k \geq 3$\label{subsec:3SwapRunningTime}}
In this subsection, we show that, by making the wrong choices, an exponential number of iterations is needed to obtain a $k$-swap optimal solution, for $k \geq 3$.
Hereto, we construct a sequence of schedules $\sigma_1, \sigma_2, \dotsc$ where $\sigma_l$ is obtained from $\sigma_{l-1}$ by swapping exactly three jobs.
The instance that we have considered consists of $m=2$ machines and four types of jobs. There is one blocking job, denoted by $\ell$, and three other types of jobs $a_i, b_i$ and $c_i (1 \leq i \leq n)$. We have $p_{a_i} = 2^{n+i+1} + 2^{i-1}$, $p_{b_i} = 2^{n+i}$, $p_{c_i} = 2^{n+i-1} + 2^{i-1}$ and $p_{\ell}$ is a sufficiently large number. We let $\ell \in M_{1}$. Therefore, $\Delta$ has a very large value, and in the procedure of obtaining a $k$-swap optimal solution, the critical machine stays the same. 

For a given schedule $\sigma$, we define a value $\omega_i(\sigma)$ based on the assignment of jobs $a_i, b_i$ and $c_i$ as follows:

\begin{center}
    $\omega_i(\sigma) = \left\{
    \begin{array}{ll}
        0 & \mathrm{if} \quad a_i \in M_{1} \quad  \mathrm{and} \quad  b_i, c_i \in M_{2}, \\
        1 & \mathrm{if} \quad a_i \in M_{2} \quad  \mathrm{and} \quad  b_i, c_i \in M_{1}, \\
        -1 & \mathrm{otherwise}.
    \end{array}
\right.$
\end{center}

We let $\omega(\sigma) = (\omega_1(\sigma), \omega_2(\sigma), \dotsc, \omega_n(\sigma))$. Whenever $\sigma$ is clear from the context, we simply write $\omega$ and $\omega_i$.
We construct the sequence of schedules such that the $\omega$ value of the schedule obtains all binary representations from $0$ to $2^n-1$.

In the initial schedule, we have $a_i \in M_{1}$ and $b_i, c_i \in M_{2}$ (for $i \in \{1, 2, ... , n\}$). Therefore, in the initial schedule, we have $\omega_i = 0$ (for all $i \in \{1, 2, ... , n\}$). 
First, we show that we are able to improve the schedule by moving $a_1$ to $M_2$, and $b_1$ and $c_1$ to $M_1$.

\begin{lemma}
\label{lem:omega1}
Suppose we have $\omega_1 = 0$ and $\ell \in M_1$. Then, $\sigma$ can be improved by making $\omega_1 = 1$.
\end{lemma}
\begin{proof}
For $\omega_1 = 0$, we have $a_1 \in M_{1}$ and $b_1, c_1 \in M_{2}$. Since $\Delta$ has a very large value, we just need to show that $p_{a_1}$ is greater than $p_{b_1} + p_{c_1}$. We have $p_{a_1} - p_{b_1} - p_{c_1} = 2^n$, which is greater than 0. Therefore, by swapping $a_1$ with $b_1$ and $c_1$, and making $\omega_1 = 1$, we improve the schedule.
\end{proof}

Next, we show that we are able to improve the schedule by moving $a_1, b_2$ and $c_2$ to $M_1$, and $b_1, c_1$ and $a_2$ to $M_2$.

\begin{lemma}
\label{lem:omega2}
Let $\sigma$ be scheduled with $\omega_1 = 1$, $\omega_2 = 0$ and $\ell \in M_1$. Then, $\sigma$ can be improved by making $\omega_1 = 0$ and $\omega_2 = 1$.
\end{lemma}
\begin{proof}
For $\omega_1 = 1$ and $\omega_2 = 0$, we have $a_2, b_1, c_1 \in M_{1}$ and $a_1, b_2, c_2 \in M_{2}$. We have $p_{a_2} - p_{a_1} - p_{b_2} = 1 > 0$ and $p_{b_1} + p_{c_1} - p_{c_2} = 2^n - 1 > 0$. Therefore, first we swap $a_2$ with $a_1$ and $b_2$, and then, we swap $b_1$ and $c_1$ with $c_2$. Note that $\Delta$ has a large value, which means that these two swaps improve the schedule. After performing these swaps, we get $\omega_1 = 0$ and $\omega_2 = 1$.
\end{proof}

Finally, we have the following lemma.

\begin{figure*}
\centering
\includegraphics[width=0.98\textwidth]{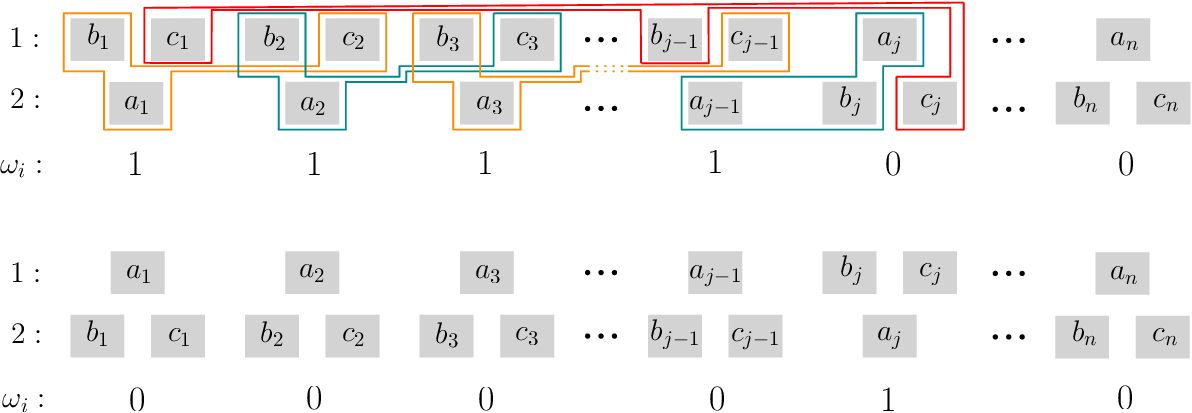}
\caption{Improving $\sigma$ by making $\omega_j = 1$ and  $\omega_i = 0$ for any $j \in \{3, ... , n\}$ and all $i \in \{1, ... , j-1\}$.}
\label{fig:omegai}
\end{figure*}

\begin{lemma}
\label{lem:omegai}
Let $j \in \{3, ... , n\}$ with $\omega_j = 0$. For all $i \in \{1, ... , j-1\}$, we have $\omega_i = 1$. Also, we have $\ell \in M_1$. Then, $\sigma$ can be improved by making $\omega_j = 1$ and  $\omega_i = 0$ for all $i \in \{1, ... , j-1\}$.
\end{lemma}
\begin{proof}
We have $b_i, c_i, a_j \in M_{1}$ and $a_i, b_j, c_j \in M_{2}$ when $\omega_j = 0$ and $\omega_i = 1$ for all $i \in \{1, ... , j-1\}$ (see Fig.~\ref{fig:omegai}).
We have $p_{a_j} - p_{a_{j-1}} - p_{b_j} = 2^{j-2} > 0$. Therefore, we are able to swap $a_j$ with $a_{j-1}$ and $b_j$. Also, we have $p_{b_k} + p_{c_{k+1}} - p_{a_k} = 2^{k-1} > 0$ for all $k \in \{1, ... , j-2\}$. Therefore, we are able to improve $\sigma$ by swapping $b_k$ and $c_{k+1}$ with $a_k$ for all $k \in \{1, ... , j-2\}$. Moreover, we have $p_{b_{j-1}} + p_{c_1} - p_{c_j} = 2^n + 1 - 2^{j-1} > 0$. So, we improve the schedule by swapping $b_{j-1}$ and $c_1$ with $c_j$ (see Fig.~\ref{fig:omegai}). After all these swaps, we get  $\omega_j = 1$ and  $\omega_i = 0$ for all $i \in \{1, ... , j-1\}$.
\end{proof}

According to Lemma~\ref{lem:omega1}, Lemma~\ref{lem:omega2}, and Lemma~\ref{lem:omegai}, we have all the combinations of zeros and ones for the tuple $\omega$. Therefore, we have the following theorem.

\threeswaplower*

\section{Computational Results}
In the previous sections, we studied the running time of searching the $k$-swap neighborhood as well as the number of iterations to obtain a local optimum solution w.r.t. this neighborhood. In this section,
we conduct computational experiments on various families of instances.
First, we delve into the empirical behavior of the running time of the randomized neighborhood search algorithm of Section~\ref{sec:SearchNeighborhood} by comparing this algorithm with the naive approach.
Subsequently, we examine the total number of iterations needed to find a $k$-swap optimal solution.
Finally, we analyze the convergence behavior of the $k$-swap local search for different values of $k$, by measuring and comparing the relative makespan improvement throughout the search process.
In our empirical analysis, we consider 9 classes of instances $C_1, C_2, \dotsc, C_9$, each of which consists of 50 instances. 
The instances of each class are generated with respect to two parameters $n$ and $m$, which are the number of jobs and machines in that class, respectively (see Table~\ref{table:1} for an overview). Moreover, for the processing time of each job, we have chosen an integer uniformly at random from the range $[1, 10^9]$.

\begin{table}[H]
\caption{An overview on the different classes of instances.}
\label{table:1}
\vspace{3mm}
\centering
\begin{tabular}{c c c c c c c c c c}
 \hline
  & $C_1$ & $C_2$ & $C_3$ & $C_4$ & $C_5$ & $C_6$ & $C_7$ & $C_8$ & $C_9$ \\ 
 \hline\hline
 $n$ & 50 & 100 & 200 & 50 & 100 & 200 & 50 & 100 & 200 \\ 
 $m$ & 2 & 2 & 2 & 5 & 5 & 5 & 10 & 10 & 10 \\
 max $k$ & 9 & 6 & 5 & 9 & 6 & 5 & 9 & 6 & 5 \\
 \hline
\end{tabular}
\end{table}

We have implemented both the randomized and naive $k$-swap operators in Python 3.9.6, and we have executed them on a MacBook Pro with an Apple M1 chip and 8 GB of memory running macOS 15.4.
Our implementation and our instances are publicly available on GitHub~\cite{gitHubImp}.

\begin{table*}[h!]
\caption{Average values of the running time (in milliseconds) of the randomized and naive operators for different values of $k$ in different classes of instances.}
\label{table:time}
\vspace{3mm}
\centering
\small
\begin{tabular}{c c c c c c c c c c c}
 \hline
  & $k$ & $C_1$ & $C_2$ & $C_3$ & $C_4$ & $C_5$ & $C_6$ & $C_7$  & $C_8$  & $C_9$ \\ 
 \hline\hline
 \textit{Randomized} & 1 & 0.2 & 0.3 & 1.0 & 0.3 & 0.5 & 1.0 & 0.4 & 0.6 & 1.2 \\ 
\textit{Naive} & 1 & $< 0.1$ & $< 0.1$ & $< 0.1$ & $< 0.1$ & $< 0.1$ & $< 0.1$ & $< 0.1$ & $< 0.1$ & 0.1 \\ 
\textit{Randomized} & 2 & 0.2 & 0.4 & 0.9 & 0.2 & 0.3 & 0.5 & 0.2 & 0.3 & 0.4 \\ 
\textit{Naive} & 2 & 0.1 & 0.2 & 1.0 & 0.1 & 0.2 & 0.5 & 0.1 & 0.1 & 0.3 \\ 
\textit{Randomized} & 3 & 0.9 & 2.9 & 13.8 & 0.4 & 1.0 & 3.0 & 0.3 & 0.6 & 1.5 \\ 
\textit{Naive} & 3 & 0.5 & 2.2 & 14.8 & 0.1 & 0.6 & 3.7 & 0.1 & 0.2 & 1.1 \\ 
\textit{Randomized} & 4 & 1.8 & 6.4 & 22.6 & 0.5 & 1.6 & 6.0 & 0.3 & 0.8 & 2.6 \\ 
\textit{Naive} & 4 & 2.6 & 24.0 & 291.1 & 0.4 & 3.6 & 41.6 & 0.1 & 0.7 & 7.4 \\ 
\textit{Randomized} & 5 & 7.4 & 46.0 & 310.9 & 1.0 & 4.7 & 30.7 & 0.4 & 1.5 & 8.3 \\ 
\textit{Naive} & 5 & 12.8 & 257.1 & 9784.2 & 1.0 & 17.0 & 422.8 & 0.2 & 1.8 & 40.2 \\ 
\textit{Randomized} & 6 & 14.4 & 99.5 & - & 1.4 & 9.3 & - & 0.5 & 2.2 & - \\ 
\textit{Naive} & 6 & 61.8 & 2741.0 & - & 2.0 & 70.4 & - & 0.3 & 3.7 & - \\ 
\textit{Randomized} & 7 & 43.2 & - & - & 2.1 & - & - & 0.5 & - & - \\ 
\textit{Naive} & 7 & 245.5 & - & - & 3.2 & - & - & 0.3 & - & - \\ 
\textit{Randomized} & 8 & 71.8 & - & - & 2.9 & - & - & 0.6 & - & - \\ 
\textit{Naive} & 8 & 851.7 & - & - & 5.3 & - & - & 0.4 & - & - \\ 
\textit{Randomized} & 9 & 191.0 & - & - & 4.0 & - & - & 0.7 & - & - \\ 
\textit{Naive} & 9 & 2424.4 & - & - & 6.3 & - & - & 0.4 & - & - \\ 

 \hline
\end{tabular}
\end{table*}

\begin{table*}[h!]
\caption{Average number of iterations of the randomized and naive operators for different values of $k$ in different classes of instances in obtaining a $k$-swap optimal solution.}
\label{table:NumOfIterations}
\vspace{3mm}
\centering
\small
\begin{tabular}{c c c c c c c c c c c}
 \hline
  & $k$ & $C_1$ & $C_2$ & $C_3$ & $C_4$ & $C_5$ & $C_6$ & $C_7$  & $C_8$  & $C_9$ \\ 
 \hline\hline
\textit{Randomized} & 1 & 1.0 & 1.0 & 1.0 & 1.0 & 1.0 & 1.0 & 1.0 & 1.0 & 1.0 \\ 
\textit{Naive} & 1 & 1.0 & 1.0 & 1.0 & 1.0 & 1.0 & 1.0 & 1.0 & 1.0 & 1.0 \\ 
\textit{Randomized} & 2 & 3.4 & 4.6 & 5.1 & 8.3 & 11.0 & 17.7 & 13.5 & 24.0 & 35.0 \\ 
\textit{Naive} & 2 & 1.3 & 1.6 & 1.4 & 2.8 & 2.4 & 2.9 & 6.7 & 6.6 & 5.4 \\ 
\textit{Randomized} & 3 & 6.5 & 7.5 & 8.8 & 20.8 & 33.3 & 59.4 & 32.3 & 72.2 & 152.6 \\ 
\textit{Naive} & 3 & 2.9 & 3.7 & 4.6 & 7.0 & 10.2 & 16.1 & 13.6 & 18.1 & 29.0 \\ 
\textit{Randomized} & 4 & 7.8 & 9.6 & 12.1 & 39.9 & 83.4 & 152.7 & 57.0 & 181.3 & 510.7 \\ 
\textit{Naive} & 4 & 4.0 & 5.5 & 6.9 & 11.7 & 18.4 & 33.3 & 21.3 & 35.1 & 58.6 \\ 
\textit{Randomized} & 5 & 10.3 & 13.6 & 14.8 & 64.3 & 152.8 & 285.3 & 68.2 & 338.8 & 1213.6 \\ 
\textit{Naive} & 5 & 5.1 & 7.1 & 9.3 & 15.2 & 27.0 & 59.0 & 26.2 & 53.5 & 108.9 \\ 
\textit{Randomized} & 6 & 12.7 & 14.8 & - & 86.3 & 242.5 & - & 81.8 & 556.3 & - \\ 
\textit{Naive} & 6 & 6.1 & 8.5 & - & 20.0 & 37.1 & - & 33.2 & 75.3 & - \\ 
\textit{Randomized} & 7 & 13.7 & - & - & 111.4 & - & - & 79.4 & - & - \\ 
\textit{Naive} & 7 & 7.1 & - & - & 24.9 & - & - & 37.8 & - & - \\ 
\textit{Randomized} & 8 & 15.8 & - & - & 130.7 & - & - & 84.2 & - & - \\ 
\textit{Naive} & 8 & 7.8 & - & - & 30.8 & - & - & 39.7 & - & - \\ 
\textit{Randomized} & 9 & 15.6 & - & - & 150.9 & - & - & 80.7 & - & - \\ 
\textit{Naive} & 9 & 8.4 & - & - & 35.6 & - & - & 41.8 & - & - \\ 
 \hline
\end{tabular}
\end{table*}

\begin{figure}[h]
\centering
\includegraphics[width=0.54\textwidth]{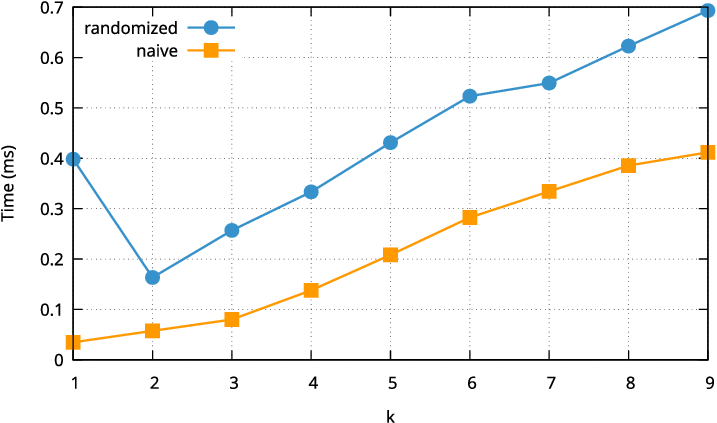}

\caption{Average running time (in milliseconds) of the randomized and naive operators for different values of $k$ in class $C_7$.}
\label{fig:diageamJ50-10}
\end{figure}

To do the experiments, we perform the randomized and naive operators for different values of $k$ on all the 450 instances that we have (see Table~\ref{table:1} for the maximum value of $k$ in each class).
To compute the average running time of the randomized and naive operators for an instance $i$ of a class $C_j$ and a value $k$, we start from an initial solution computed by the longest processing time first ($LPT$) algorithm~\cite{graham1969bounds}, and then, we divide the total running time by the number of iterations in finding a $k$-swap optimal solution with respect to these two operators.
Note that according to Lemma \ref{lem:probFailureOfKSwap}, when we are not able to find an improving solution by the randomized $k$-swap algorithm, we need to run it at most $\dfrac{2^k}{\binom{k}{\ceil{\frac{k}{2}}}}$ times to amplify the probability of success of the algorithm.
Finally, to calculate the average running time of the randomized and naive operators for a class $C_j$ and a value $k$, we compute the average value of the average running times of the randomized and naive operators over all the instances in class $C_j$ for the value $k$, respectively. 
For our detailed results, see Table~\ref{table:time} and Table~\ref{table:NumOfIterations}.

We observe that in most of the classes by increasing the value of $k$, the randomized operator is capable of finding an improving neighbor faster than the naive one (see Table~\ref{table:time}). However, in some cases such as in class $C_{7}$, the running time of the naive operator is less than the randomized operator even for larger values of $k$ (see Fig.~\ref{fig:diageamJ50-10}). More precisely, in classes where we have relatively more jobs on each machine, the randomized operator is able to find an improving neighbor faster than the naive one. Also, when we increase the value of $k$, the gap between the running times of these two operators increases (see Fig.~\ref{fig:diageamJ100}). 
As a matter of fact, for finding an improving neighbor in the $k$-swap neighborhood, the naive operator only searches the neighborhood, however, the randomized operator performs some preprocessing first and after that searches a smaller neighborhood. Therefore, when we have relatively fewer jobs on each machine, it is better to use the naive operator to find an improving neighbor, and when we have more jobs, by using the randomized operator, we are able to find a solution faster.

\begin{figure}[h!]
\centering
\includegraphics[width=0.48\textwidth]{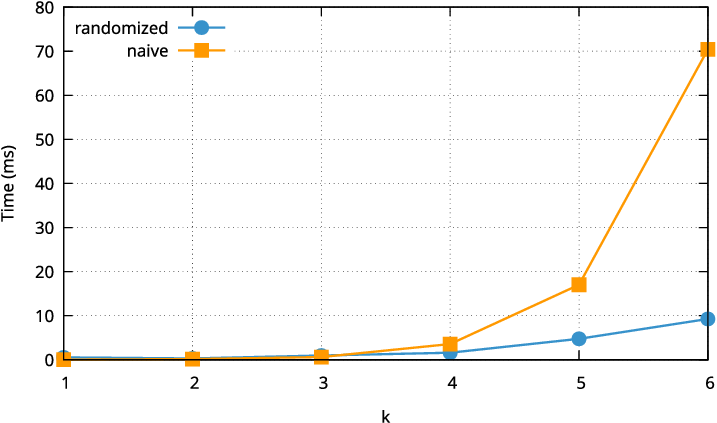}
\includegraphics[width=0.48\textwidth]{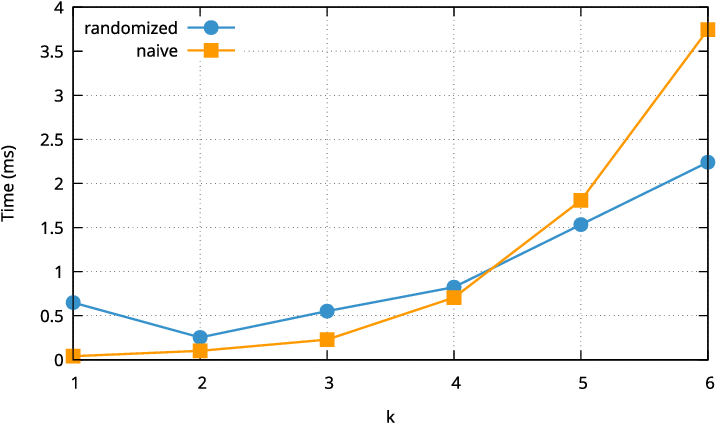}

\caption{Average running time (in milliseconds) of the randomized and naive operators for different values of $k$ in classes $C_5$ (left diagram) and $C_8$ (right diagram).}
\label{fig:diageamJ100}
\end{figure}

In Theorem~\ref{theorem:2-swapOptSolUpperBound}, we provide a polynomial upper bound for the number of iterations in finding a 2-swap optimal solution when we have two machines in our schedule. 
According to our computational results, we do not observe any significant increase in the number of iterations when the number of machines is 5 or 10 (see Table~\ref{table:NumOfIterations}).
This computational result raises the question of whether it is possible to establish a polynomial upper bound on the number of iterations when dealing with schedules involving more than two machines.
Moreover, in Theorem~\ref{theorem:lowerbound}, we establish an exponential lower bound for the case of $k \geq 3$. 
The computational experiments in Table~\ref{table:NumOfIterations} lead us to question whether the establishment of this lower bound reflects an exceptionally pessimistic scenario. In other words, the occurrence of such a scenario might be rare in practice.

\begin{figure}[h!]
\centering
\includegraphics[width=0.48\textwidth]{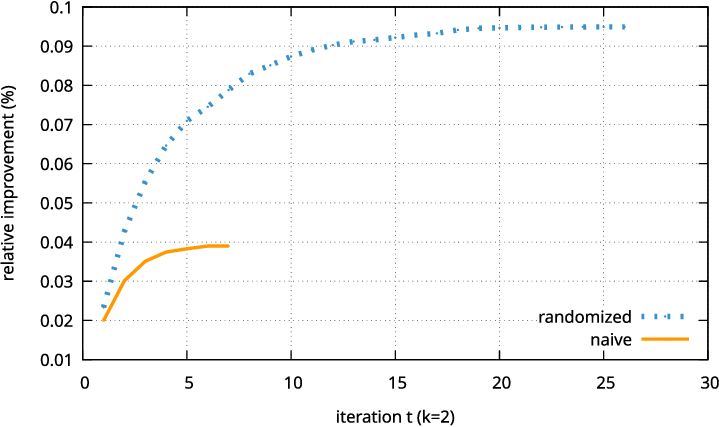}
\includegraphics[width=0.48\textwidth]{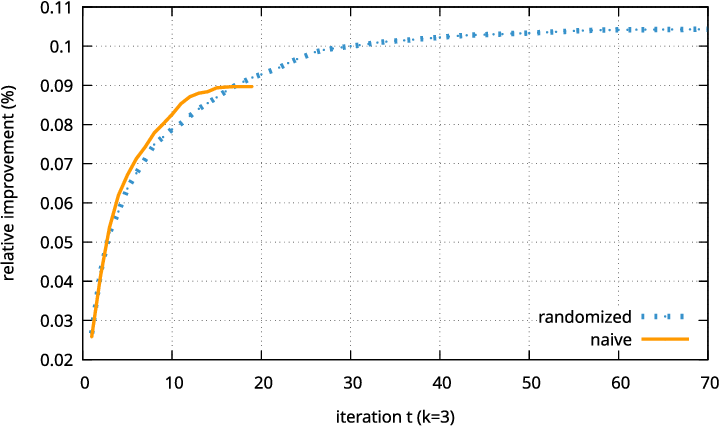}
\caption{Average values of the relative improvement in makespan (in percentage) of the randomized and naive operators per iteration~$t$ for $k=2$ (left diagram) and $k=3$ (right diagram) in class $C_5$.}
\label{fig:convergenceC51}
\end{figure}

\begin{figure}[h!]
\centering
\includegraphics[width=0.48\textwidth]{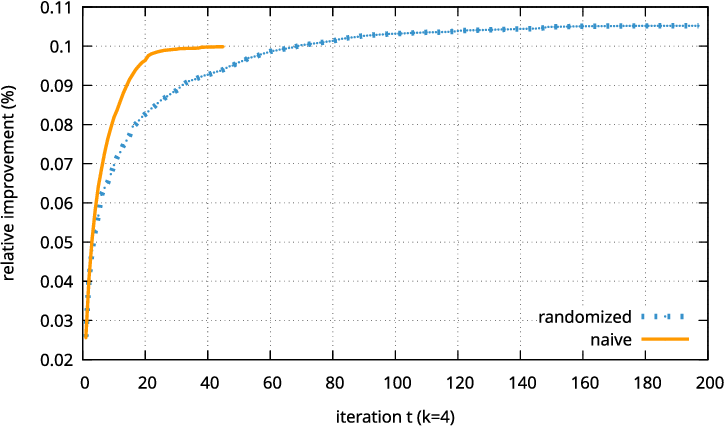}
\includegraphics[width=0.48\textwidth]{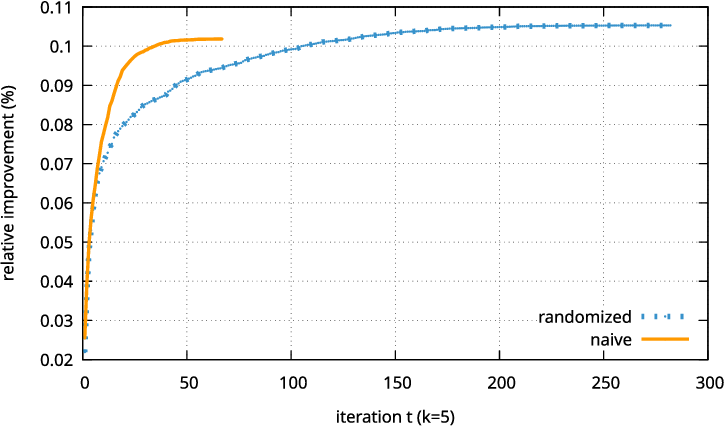}
\caption{Average values of the relative improvement in makespan (in percentage) of the randomized and naive operators per iteration~$t$ for $k=4$ (left diagram) and $k=5$ (right diagram) in class $C_5$.}
\label{fig:convergenceC52}
\end{figure}

Finally, we analyze the convergence behavior of the $k$-swap neighborhood over different values of $k$. To do so, we measure the relative improvement in makespan at each iteration $t$, computed as $1 - \frac{L_{\max}(t)}{L_{\max}(0)}$, where $L_{\max}(t)$ is the makespan at iteration $t$, and $L_{\max}(0)$ is the makespan of the initial solution. A higher value indicates a greater improvement.
The complete data on relative makespan improvement across all instances and values of $k$ is available on our GitHub page~\cite{gitHubImp}.
As an illustrative example, Figures~\ref{fig:convergenceC51} and~\ref{fig:convergenceC52} display the convergence curves for class $C_5$, showing that larger
$k$-values typically lead to greater makespan improvements.

\begin{table*}[h!]
\caption{Average value of the relative improvement in makespan in the $k$-swap optimal solution (in percentage) of the randomized and naive operators for different values of $k$ in different classes of instances.}
\label{table:finalMSI}
\vspace{3mm}
\centering
\fontsize{8.5}{10.5}\selectfont{
\begin{tabular}{c c c c c c c c c c c}
 \hline
  & $k$ & $C_1$ & $C_2$ & $C_3$ & $C_4$ & $C_5$ & $C_6$ & $C_7$  & $C_8$  & $C_9$ \\ 
 \hline\hline
\textit{Randomized} & 1 & 0.0 & 0.0 & 0.0 & 0.0 & 0.0 & 0.0 & 0.0 & 0.0 & 0.0 \\ 
\textit{Naive} & 1 & 0.0 & 0.0 & 0.0 & 0.0 & 0.0 & 0.0 & 0.0 & 0.0 & 0.0 \\ 
\textit{Randomized} & 2 & 0.04915 & 0.01564 & 0.00406 & 0.44485 & 0.09683 & 0.02879 & 1.74698 & 0.45392 & 0.11045 \\ 
\textit{Naive} & 2 & 0.0156 & 0.00669 & 0.00207 & 0.24151 & 0.03976 & 0.01542 & 1.29887 & 0.25511 & 0.0528 \\ 
\textit{Randomized} & 3 & 0.05334 & 0.01629 & 0.00413 & 0.50956 & 0.1064 & 0.03018 & 2.12856 & 0.51444 & 0.11969 \\ 
\textit{Naive} & 3 & 0.04539 & 0.01537 & 0.00404 & 0.41096 & 0.09148 & 0.02833 & 1.7008 & 0.41539 & 0.107 \\ 
\textit{Randomized} & 4 & 0.05365 & 0.01632 & 0.00413 & 0.51875 & 0.10731 & 0.03024 & 2.21709 & 0.52333 & 0.12044 \\ 
\textit{Naive} & 4 & 0.05092 & 0.01618 & 0.00413 & 0.46609 & 0.10186 & 0.02958 & 1.96538 & 0.4751 & 0.11367 \\ 
\textit{Randomized} & 5 & 0.0537 & 0.01632 & 0.00413 & 0.52172 & 0.1074 & 0.03024 & 2.24093 & 0.52511 & 0.12051 \\ 
\textit{Naive} & 5 & 0.05269 & 0.01629 & 0.00413 & 0.4851 & 0.10387 & 0.02999 & 1.99711 & 0.49794 & 0.1182 \\ 
\textit{Randomized} & 6 & 0.0537 & 0.01632 & - & 0.52214 & 0.10742 & - & 2.23807 & 0.52546 & - \\ 
\textit{Naive} & 6 & 0.05325 & 0.01631 & - & 0.49787 & 0.10582 & - & 2.08486 & 0.50736 & - \\ 
\textit{Randomized} & 7 & 0.0537 & - & - & 0.52246 & - & - & 2.22278 & - & - \\ 
\textit{Naive} & 7 & 0.05349 & - & - & 0.50892 & - & - & 2.11134 & - & - \\ 
\textit{Randomized} & 8 & 0.0537 & - & - & 0.52246 & - & - & 2.23074 & - & - \\ 
\textit{Naive} & 8 & 0.05361 & - & - & 0.51179 & - & - & 2.12059 & - & - \\ 
\textit{Randomized} & 9 & 0.0537 & - & - & 0.52249 & - & - & 2.23129 & - & - \\ 
\textit{Naive} & 9 & 0.05365 & - & - & 0.51522 & - & - & 2.13363 & - & - \\ 
 \hline
\end{tabular}}
\end{table*}

This trend is also reflected in the average improvement of the makespan of the $k$-swap optimal solution, summarized for each $k$ in Table~\ref{table:finalMSI}, which confirms that larger neighborhoods lead to better makespan values.
This behavior is aligned with theoretical insights: since the set of local optima for $k$-swap is a subset of that for $(k-1)$-swap, larger values of $k$ help escape inferior local optima. While the number of iterations might increase with $k$ (see Table~\ref{table:NumOfIterations}), this reflects continued improvement rather than inefficiency. Overall, the analysis suggests that larger values of $k$ yield better-quality solutions in terms of makespan improvement, without introducing significant additional computational cost.

\section{Concluding Remarks}

The main focus of this paper is to analyze the $k$-swap neighborhood for the problem of makespan minimization on identical parallel machines, both theoretically and empirically, providing insights into how well the basic implementation of this neighborhood performs.
To find an improving solution in this neighborhood, we have proposed a randomized algorithm that finds such an improving neighbor, if it exists, in $O(n^{\ceil{\frac{k}{2}} + O(1)})$ steps. This approach is faster than the naive approach, which requires $O(k \cdot n^k)$ steps. 
Moreover, we have established an almost matching lower bound of $\Omega(n^{\ceil{\frac{k}{2}}})$ on the time needed to find an improving neighbor, assuming the $k$-sum conjecture.

Additionally, we have analyzed the total number of iterations required to converge to a $k$-swap optimal solution.
For $k=2$, we have provided a polynomial upper bound of $O(n^4)$ when we have only two machines in our schedule. 
Our computational results indicate that there is no significant increase in the number of iterations when the number of machines is more than two. This observation raises the question of whether a polynomial upper bound on the number of iterations can also be established for schedules with more than two machines, which remains as an open problem.

On the other hand, for $k \geq 3$, we have shown that there exists a sequence of $2^{\Omega(n)}$ improving $3$-swaps before a $k$-swap optimal solution is reached.
The computational experiments lead us to question whether the established lower bound represents an exceptionally pessimistic scenario which might be rare in practice. 
A similar situation arises in the case of the 2-opt neighborhood for TSP.
More specifically, Englert et al.~\cite{englert2014worst} construct Euclidean TSP instances where 2-opt requires an exponential number of steps to reach a locally optimal tour. However, despite this pessimistic result, 2-opt performs reasonably well in practice~\cite{aarts2003local}.
To address such discrepancies between worst-case and practical performance, the notion of \emph{smoothed analysis}, a hybrid of worst-case and average-case, was introduced by Spielman and Teng~\cite{spielman2004smoothed}, aiming to provide a more realistic measure of an algorithm's performance. 
In a follow-up work~\cite{rohwedder2024smoothed}, we have shown that the smoothed number of $k$-swaps needed to find a $k$-swap optimal solution is bounded from above by a polynomial function in $n$ and $m$ as well as some smoothing parameter.

The focus of the empirical results in this paper is to compare the running time of the randomized and naive implementations for finding an improving $k$-swap neighbor. Therefore, we have used the most basic form of local search, known as iterative improvement. 
According to these results, in instances where we have relatively more jobs compared to the number of machines, the randomized implementation performs faster than the naive approach (especially for larger values of $k$).
Metaheuristics like simulated annealing~\cite{aarts1989simulated} and tabu search~\cite{glover1989tabu, glover1990tabu} will most likely lead to improved solutions, but since they allow for non-improving neighbors, they are not suitable for this comparison.
Naturally, $k$-swap neighborhood can also serve as a foundation for these metaheuristics, which are designed to escape poor local optima.

Note that any $k$-swap neighborhood includes all jump neighbors, and therefore, the performance guarantee of a jump optimal solution~\cite{finn1979linear} carries over the $k$-swap neighborhood, that is, the solution value is bounded by $2 - \frac{2}{m+1}$ times the optimal solution value.

Lastly, since not always all machines are identical in practice, future research could focus on extending this work to uniform or unrelated machine scheduling problems.

\bibliographystyle{plain} 
\bibliography{cas-refs}

\end{document}